\documentclass[a4paper,UKenglish,cleveref, autoref]{oasics-v2019}

\usepackage{lmodern}
\usepackage{pythonhighlight}
\usepackage{hyperref}

\newcommand{\polylog}{\textup{polylog}}

\title{Fast and Simple Modular Subset Sum}

\author{Kyriakos Axiotis}{Massachusetts Institute of Technology}{}{}{}%
\author{Arturs Backurs}{Toyota Technological Institute at Chicago}{}{}{}%
\author{Karl Bringmann}{Saarland University and Max-Planck Institute for Informatics, Saarland Informatics Campus}{}{}{}
\author{Ce Jin}{Massachusetts Institute of Technology}{}{}{}%
\author{Vasileios Nakos}{Saarland University and Max-Planck Institute for Informatics, Saarland Informatics Campus}{}{}{}
\author{Christos Tzamos}{University of Wisconsin-Madison}{}{}{}%
\author{Hongxun Wu}{Institute for Interdisciplinary Information Sciences, Tsinghua University, China}{}{}{}%

\authorrunning{K.\,Axiotis, A.\,Backurs, K.\ Bringmann, C.\,Jin, V.\ Nakos, C.\,Tzamos and H.\, Wu}%

\ccsdesc[100]{Theory of computation~Algorithm design techniques}

\keywords{Modular Subset Sum, rolling hash, dynamic strings}%

\category{}%

\relatedversion{}%

\supplement{}%

\funding{\emph{Karl Bringmann and Vasileios Nakos:} This work is part of the project TIPEA that has received funding from the European Research Council (ERC) under the European Unions Horizon 2020 research and innovation programme (grant agreement No.~850979). \emph{Arturs Backurs:} Supported by an NSF Grant CCF-2006806.}%

\nolinenumbers %

\hideOASIcs  %

\EventEditors{John Q. Open and Joan R. Access}
\EventNoEds{2}
\EventLongTitle{42nd Conference on Very Important Topics (CVIT 2016)}
\EventShortTitle{CVIT 2016}
\EventAcronym{CVIT}
\EventYear{2016}
\EventDate{December 24--27, 2016}
\EventLocation{Little Whinging, United Kingdom}
\EventLogo{}
\SeriesVolume{42}
\ArticleNo{23}

\usepackage{listings}
\usepackage{algorithm}
\usepackage[noend]{algpseudocode}

\renewcommand{\epsilon}{\varepsilon}

\renewcommand\bar\overline
\global\long\def\Oh{{\cal O}}
\global\long\def\tOh{\widetilde{\Oh}}

\begin{document}

\maketitle

\begin{abstract}
We revisit the Subset Sum problem over the finite cyclic group $\mathbb{Z}_m$ for some given integer $m$. A series of recent works has provided near-optimal algorithms for this problem under the Strong Exponential Time Hypothesis. Koiliaris and Xu (SODA'17, TALG'19) gave a deterministic algorithm running in time $\tOh(m^{5/4})$, which was later improved to $\Oh(m \log^7 m)$ randomized time by Axiotis et al.~(SODA'19).

In this work, we present two simple algorithms for the Modular Subset Sum problem running in near-linear time in $m$, both efficiently implementing Bellman's iteration over $\mathbb{Z}_m$. The first one is a \emph{randomized} algorithm running in time $\Oh(m\log^2 m)$, that is based solely on rolling hash and an elementary data-structure for prefix sums; to illustrate its simplicity we provide a short and efficient implementation of the algorithm in Python. Our second solution is a \emph{deterministic} algorithm running in time $\Oh(m \,\polylog\, m)$, that uses dynamic data structures for string manipulation. 

We further show that the techniques developed in this work can also lead to simple algorithms for the All Pairs Non-Decreasing Paths Problem (APNP) on undirected graphs, matching the near-optimal running time of $\tOh(n^2)$ provided in the recent work of Duan et al.~(ICALP'19).
\end{abstract}

\section{Introduction}
In the \emph{Subset Sum} problem, one is given a multiset $X=\{x_1,x_2,\ldots,x_n\}$ of integers along with an integer target $t$, and is asked to decide if there exists a subset of $X$ that sums to the target $t$. In the \emph{Modular Subset Sum} generalization of the problem, all sums are taken over the finite cyclic group $\mathbb{Z}_m$ for some given integer $m$. 

Subset Sum is a fundamental problem in Computer Science known to be NP-complete but only weakly as it admits pseudo-polynomial time algorithms. In particular, the Dynamic Programming algorithm of Bellman \cite{Bellman57} solves the problem in $\Oh(nt)$ time. It works by iteratively computing all attainable subset sums when using only the first $i$ integers. More specifically, it starts with $S^0 = \{0\}$ and computes $S^{i}$ as $S^{i-1} \cup (S^{i-1} + x_{i})$, where $S^{i-1}+x_i = \{s+x_i\ |\ s \in S^{i-1}\}$.

The above algorithm can be straightforwardly applied to give an\footnote{All our running time bounds assume that the usual arithmetic operations on $\log(m)$-bit numbers can be performed in constant time.} $\Oh(nm)$ time algorithm for the modular case. Recent work by Koiliaris and Xu~\cite{KX2019} obtained an improved deterministic algorithm running in\footnote{After an $\Oh(n+m)$-time preprocessing we can assume that $n = \Oh(m)$, see Section~\ref{sec:prelims}. After this preprocessing, we can express the running time in terms of $m$ only. We ignore the preprocessing time in most running time bounds stated in this paper; this only hides an additive $\Oh(n)$.} $\tOh(m^{5/4})$ that relies on structural results from number theory~\cite{hamidoune2008complete}. 
A follow up work by Axiotis et al.~\cite{AxiotisBT18} presented a randomized algorithm that improves the running time to $\Oh(m \log^7m)$ using ideas based on linear sketching. The obtained running time matches (up to subpolynomial factors) the conditional lower bound of Abboud et al.~\cite{ABHS17} based on the Strong Exponential Time Hypothesis which implies that no $\Oh(m^{1-\epsilon})$ algorithms exist for any constant $\epsilon > 0$. %Another near-linear time algorithm follows from~\cite{BringNak21}, where the authors show that given sets $A_1,A_2,\ldots,A_n \subseteq \mathbb{Z}_m$, one can compute the sumset $A_1+A_2+\ldots A_n$ in near-linear output-sensitive time; Modular Subset Sum is a special case where all $A_i$ are of size $2$.

While prior work obtained near-optimal algorithms for Modular Subset Sum, the resulting algorithms are complex and their analysis is relatively involved. In this work, we present two simple near-optimal algorithms. 
Our simplest algorithm (see Section~\ref{sec:algoone}) is \emph{randomized} and runs in time $\Oh(m \log^2 m)$. More precisely, the algorithm produces the whole set $X^\ast$ of attainable subset sums of the multiset $X$ in time $\Oh(|X^\ast| \log^2 m)$. 
The idea behind our algorithm is a fast implementation of Bellman's iteration and requires only two elementary techniques, rolling hashing and a data structure for maintaining prefix sums. These techniques are already taught in undergraduate level algorithms classes. We believe that our simple algorithm can serve as an example application when these techniques are introduced.

Our second algorithm (see Section~\ref{sec:algotwo}) is \emph{deterministic} and solves Modular Subset Sum in time $\tOh(m) = \Oh(m \,\polylog\, m)$. More precisely, the algorithm produces the set~$X^\ast$ of attainable subset sums in time\footnote{After an $\Oh(n \log n)$-time preprocessing we can assume that $n = \Oh(|X^*|)$, see Section~\ref{sec:prelims}. We ignore this preprocessing in our output-sensitive running time bounds; this only hides an additive $\Oh(n \log n)$.} $\tOh(|X^\ast|) = \Oh(|X^\ast| \,\polylog\, |X^\ast|)$. 
This algorithm is based on a classic data structure for string manipulation, and apart from this data structure the algorithm is simple. 
The idea of solving Modular Subset Sum via dynamic string data structures has already been suggested in~\cite{AxiotisBT18}, however, the algorithm proposed in~\cite{AxiotisBT18} runs in time $\Oh(|X^\ast| \,\polylog\, m)$, which we improve to $\Oh(|X^\ast| \,\polylog\, |X^\ast|)$.

\paragraph*{Techniques for the First Algorithm}
We first explain the technical innovation behind our randomized $\Oh(m \log^2 m)$ algorithm (Theorem~\ref{thm:result_1} in Section~\ref{sec:algoone}). At the core of our argument is a new method for computing the symmetric difference $S_1 \triangle S_2$ between two sets $S_1,S_2 \subseteq [m]$ in output-sensitive time upon specific updates on those two sets. The idea is to use hashing to compare the indicator vectors of the two sets. If the two hashes are the same, then the two sets are the same w.h.p. If not, we compute the symmetric difference of the sets $S_1$ and $S_2$ by recursing on the first and the second half of the universe, $\{1,\dots,\lceil m/2 \rceil\}$ and $\{\lceil m/2 \rceil + 1,\dots, m\}$. In total, at most $\log m + 1$ hashes need to be computed per element of $S_1 \triangle S_2$.
Each hash that needs to be computed corresponds to a contiguous interval of the indicator vectors. It can be evaluated in $\Oh(\log m)$ time given access to a data structure that maintains prefix sums of a polynomial rolling hash function for the indicator vectors of each of the sets.

We show that this idea can be applied to other problems beyond Modular Subset Sum. In particular, we consider the problem of all-pairs non-decreasing paths (APNP) in undirected graphs, where we obtain near-optimal running time $\Oh(n^2 \log n)$ improving the state of the art for this problem, see Appendix~\ref{app:apnp}.

These two algorithms for Modular Subset Sum and APNP are simple to describe and to analyze. To illustrate their simplicity, we provide short but detailed implementations in Python for both algorithms in the appendix (see Appendix~\ref{app:implementation} and~\ref{app:apnpimplementation}).

\paragraph*{Techniques for the Second Algorithm}
Now let us describe our deterministic $\Oh(m \,\polylog\, m)$ algorithm (Theorem~\ref{thm:result_2} in Section~\ref{sec:algotwo}). 
The core of this algorithm is again a fast method for computing the symmetric difference $S_1 \triangle S_2$ for sets $S_1,S_2 \subseteq [m]$. Consider the indicator vectors of $S_1$ and $S_2$ and interpret them as length-$m$ strings $z_1,z_2$ over alphabet $\{0,1\}$. Then the symmetric difference $S_1 \triangle S_2$ corresponds to all positions at which the strings $z_1,z_2$ differ. We thus obtain the first element of the symmetric difference by computing the longest common prefix of $z_1$ and $z_2$. Generalizing this idea, we can enumerate the symmetric difference using one longest common prefix query per output element. We implement such queries by using a classic data structure for dynamically maintaining a family of strings under concatenations, splits, and equality tests due to Mehlhorn et al.~\cite{MehlhornSU97}.

Implementing this idea naively leads to a running time of $\Oh(|X^\ast| \,\polylog\, m)$. By working on the run-length encoding of the strings $z_1,z_2$, we further improve the running time to $\Oh(|X^\ast| \,\polylog\, |X^\ast|)$.

\paragraph*{Further Related Work}

In addition to Modular Subset Sum, there has recently been a lot of interest in obtaining faster algorithms for other related problems, like non-modular Subset Sum~\cite{Bringmann17, KX2019, DBLP:conf/soda/JinW19} and Knapsack~\cite{rhee2015faster,DBLP:conf/soda/Chan18a, DBLP:conf/icalp/Jin19, bateni2018fast, DBLP:conf/icalp/AxiotisT19, jansen2018integer, eisenbrand2018proximity}, and providing conditional lower bounds~\cite{ABHS17, cygan2019problems, kunnemann2017fine}.

\section{Preliminaries}
\label{sec:prelims}

Let $X$ be a multiset of integers in $\mathbb{Z}_m$. Recall that we denote by $X^*$ the set of all attainable subset sums of $X$ modulo $m$.
In this section we present a preprocessing that ensures $n = \Oh(|X^*|)$, and thus also $n = \Oh(m)$, see Lemma~\ref{lem:preprocessing}. This is inspired by a similar preprocessing by Koiliaris and Xu~\cite[Lemma 2.4]{KX2019}.

For any $x \in \mathbb{Z}_m$ we write $\mu_X(x)$ for the multiplicity of $x$ in $X$, that is, how often $x$ appears in the multiset $X$. Note that the cardinality $|X|$ is equal to the total multiplicity $\sum_{x=0}^{m-1} \mu_X(x)$.

\begin{lemma}[Cf. Lemma 2.3 in \cite{KX2019}] \label{lem:preprocessing_onestep}
  Let $x \in X$ with $\mu_X(x) \ge 3$. Consider the multiset~$Y$ resulting from removing two copies of $x$ from $X$ and adding the number $2x \bmod m$ to it. Then $X^* = Y^*$ and $|Y| = |X| - 1$.
\end{lemma}
\begin{proof}
  Clearly, any subset sum of $Y$ modulo $m$ is also a subset sum of $X$ modulo $m$. In the other direction, for any subset of $X$ containing at least two copies of $x$ we can replace two of these copies by one copy of $2x \bmod m$, thereby transforming it into a subset of $Y$ with the same sum modulo $m$. This proves $X^* = Y^*$. The cardinality $|Y| = |X| - 1$ is immediate.
\end{proof}

\begin{lemma} \label{lem:preprocessing}
  Given a multiset $X$ of $n$ integers in $\mathbb{Z}_m$, in time $\Oh(\min\{n \log n, n + m\})$ we can compute a multiset $Y$ over $\mathbb{Z}_m$ such that $Y^* = X^*$ and $|Y| \le \min\{n, 2|X^*|\}$. 
\end{lemma}

Note that in particular $|Y| \le \min\{n, 2m\}$.

\begin{algorithm}[ht]
\caption{Single Step of the Preprocessing}
\begin{algorithmic}[1]
\Function{Preprocessing-Check}{$x$}
	\If {$\mu_X(x) \ge 3$} 
	  \State $\mu_X(x) \mathrel{-}= 2$
	  \State $\mu_X(2x \bmod m) \mathrel{+}= 1$
	  \State $\textproc{Preprocessing-Check}(x)$
	  \State $\textproc{Preprocessing-Check}(2x \bmod m)$
    \EndIf
\EndFunction
\end{algorithmic}
\label{a:findnewsums}
\end{algorithm}
  
\begin{proof}
  We exhaustively apply Lemma~\ref{lem:preprocessing_onestep} by calling $\textproc{Preprocessing-Check}(x)$ on each $x \in X$. 
  After these calls have ended, Lemma~\ref{lem:preprocessing_onestep} is no longer applicable, and thus the resulting multiset $Y$ satisfies $\mu_Y(x) \le 2$ for all $x \in \mathbb{Z}_m$. By Lemma~\ref{lem:preprocessing_onestep} we have $Y^* = X^*$ and thus the support of $Y$ is a subset of $X^*$, which implies $|Y| \le 2 |X^*|$. The inequality $|Y| \le n$ is trivial, since the cardinality only decreases throughout this algorithm.
  
  Since each successful check decreases the cardinality, there are $\Oh(n)$ successful checks. Since each successful check calls two additional checks, there are $\Oh(n)$ unsuccessful checks. It follows that the procedure $\textproc{Preprocessing-Check}$ is called $\Oh(n)$ times in total. 
  
  It remains to argue in which running time one call of $\textproc{Preprocessing-Check}$ can be implemented. An easy solution is to store an array $M$ of length $m$ such that $M[x] = \mu_X(x)$. Initializing $M$ takes time $m$. One call of $\textproc{Preprocessing-Check}$ can then easily be implemented in time $\Oh(1)$.  This yields total time $\Oh(n+m)$.
  
  In order to avoid time (or space) $\Oh(m)$, we can alternatively store all distinct elements of~$X$ in a balanced binary search tree $\cal T$, and store the number $\mu_X(x)$ at the node corresponding to $x$ in $\cal T$. Building this tree initially takes time $\Oh(n \log n)$, for sorting the set $X$. One call of $\textproc{Preprocessing-Check}$ can then be implemented in time $\Oh(\log n)$, resulting in a total running time of $\Oh(n \log n)$.  
\end{proof}

We hence assume $n \le 2 |X^\ast| \le 2m$ in the remainder of this paper.

\section{Algorithm I: Rolling Hash and Polynomial Identity Testing}
\label{sec:algoone}

To describe our implementation, we consider Bellman's iteration\footnote{Here and in the remainder of this paper, we write $S+x := \{s+x \bmod m \mid s \in S\}$, for a given modulus~$m$.} $S^{i} = S^{i-1} \cup (S^{i-1} + x_{i})$. 
Our goal is to compute $X^\ast = S^{n}$, that is, the set of all attainable modular subset sums.
Note that, if we could efficiently compute the new sums $C^i \triangleq (S^{i-1} + x_{i}) \setminus S^{i-1}$, we would be able to implement the Bellman interation as $S^i = S^{i-1} \cup C^i$. We will shortly show that $C^i$ can be computed in output-sensitive time $\Oh((|C^i|+1)\log^2 m)$. This implies that the total time to evaluate $S^n = C^1 \cup \ldots \cup C^n$ is $\Oh((|C^1|+1)\log^2 m) + \ldots + \Oh((|C^n|+1)\log^2 m) \leq \Oh((|X^\ast|+n) \log^2 m) \leq \Oh(|X^\ast| \log^2 m) \leq \Oh(m \log^2 m)$ since the sets $C^i$ are disjoint and their union is of size $|X^\ast| \le m$.

We now argue how to compute these new subset sums, $(S^{i-1} + x_{i}) \setminus S^{i-1}$, efficiently. Instead of considering the set difference between the sets $S^{i-1} + x_{i}$ and $S^{i-1}$, we will consider their symmetric difference $(S^{i-1} + x_{i}) \triangle S^{i-1} = C^i \cup D^i$, where $D^i = S^{i-1} \setminus (S^{i-1} + x_{i})$. An important observation made in \cite{AxiotisBT18} is that since the sets $S^{i-1}$ and $S^{i-1} + x_{i}$ have the same size, the symmetric difference will have size exactly $2 |C^i|$ as $|C^i| = |D^i|$. Thus, recovering this larger set $C^i \cup D^i$ in output sensitive time is asymptotically the same as recovering only the elements of $C^i$. For notational convenience, we call elements of $D^i$ ``ghost sums'' in the sense that they are not new subset sums.

We now provide a recursive function (\cref{a:findnewsums}) that given a set of integers $S$ and integers $a,b,x \in \{0,\ldots,m\}$ computes $((S+x)\setminus S)\cap [a,b)$. Calling the function with $S=S^{i-1}$, $a=0$, $b=m$ and $x=x_i$, we can recover $C^i$. We will show that the function outputs $C^i$ in time $\Oh((|C^i|+1)\log^2 m)$, which is what we need.

\begin{algorithm}[ht]
\caption{Find new subset sums in range $[a,b)$}
\begin{algorithmic}[1]
\Function{Find-New-Sums}{$a, b, x, S$}
	\If {$(S+x) \cap [a,b) = S \cap [a,b)$} \Return $\emptyset$
    \EndIf
    \If {$b = a+1$} 
    	\If {$a \in (S + x) \setminus S$}
        	\Return $\{a\}$  \Comment{$a$ is a new subset sum}
        \Else
        	\ \Return $\emptyset$  \Comment{$a \in S \setminus (S + x)$ is a ghost sum}
        \EndIf
	\Else 
	    \ \Return $\textproc{Find-New-Sums}(a, \lfloor\frac{a+b}{2}\rfloor, x, S) \cup \textproc{Find-New-Sums}(\lfloor\frac{a+b}{2}\rfloor, b, x, S)$
    \EndIf
\EndFunction
\end{algorithmic}
\label{a:findnewsums}
\end{algorithm}

We implement the function efficiently by maintaining a data structure for the characteristic vector of the set $S$ that allows efficient membership queries, updates and equality checks between different parts of the vector as required in line 2.

We interpret the set $S$ as a characteristic vector and write $S_i = 1$ if $i \in S$ and $S_i=0$ if $i \not\in S$. We also extend this notation to $i < 0$ or $i \ge m$ by setting $S_i = S_{i \bmod m}$. To check that $(S+x) \cap [a,b) = S \cap [a,b)$, we need to check that the binary sequences $(S+x)_a, \ldots, (S+x)_{b-1}$ and $S_a, \ldots, S_{b-1}$ are equal. To check the equality of the two sequences we will use polynomial identity testing. In particular, let $r$ be a uniformly random integer from $\{0, \ldots, p-1\}$ for a large enough prime $p$ (which we will choose later). Then, with high probability, it is sufficient to check that $\sum_{i=a}^{b-1} (S+x)_i r^i = \sum_{i=a}^{b-1} S_i r^i\ (\text{mod}\ p)$ to conclude the equality of the sequences. The latter condition is equivalent to $r^x \sum_{i=a-x}^{b-x-1} S_i r^i \equiv \sum_{i=a}^{b-1} S_i r^i\ (\text{mod}\ p)$, which we can rearrange to $\sum_{i=m+a-x}^{m+b-x-1} S_i r^i \equiv r^{m-x} \sum_{i=a}^{b-1} S_i r^i\ (\text{mod}\ p)$.
This is the same as $f(m+b-x)-f(m+a-x) \equiv r^{m-x} (f(b)-f(a))\ (\text{mod}\ p)$, where $f(t)\triangleq \sum_{i=0}^{t-1} S_i r^i \bmod p$ for all $t=0, \ldots, 2m$.

\paragraph*{Correctness} To argue the correctness, we observe that for any two binary sequences $x,y \in \{0,1\}^t$, prime $p$ and a random integer $r \in \{0, \ldots, p-1\}$ we have $\Pr[\sum_i x_i r^i = \sum_i y_i r^i\ (\text{mod}\ p)]\leq t/p$ if $x \neq y$ and $\Pr[\sum_i x_i r^i = \sum_i y_i r^i\ (\text{mod}\ p)] = 1$ if $x=y$. This is also known as the Rabin-Karp rolling hash function. Choosing $p = \Theta(m^2 \log (m) / \delta)$, suffices to have the algorithm fail with probability at most $\delta$. 
This is because a single randomized comparison fails with probability at most $\frac {\delta} {m \log m}$ and by a union bound the probability that any of the $m \log m$ comparisons performed in the algorithm fails is at most $\delta$. Assuming basic arithmetic operations between $\Oh(\log m)$-bit numbers take constant time, we can choose $\delta = 1/\text{poly}(m)$ to obtain a high probability of success. 

\paragraph*{Running Time} We will show that the prefix sums $f(t)=\sum_{i=0}^{t-1} S_i r^i \ (\text{mod}\ p)$ can be evaluated in time $\Oh(\log n)$, which will lead to the required running time for computing $C^i$ as we will see later. Additionally, we need that the data structure can update the characteristic vector of the set~$S$ in $\Oh(\log n)$ time (to be able to implement the Bellman iteration $S^i = S^{i-1}\cup C^i$ efficiently). These requirements can be abstracted as follows. We have a sequence of integers $T_0, \ldots, T_{m-1}$ and in each step we either want to compute the prefix sum $g(t) \triangleq \sum_{i=0}^{t-1} T_i$ for some integer $t \in \{0, \ldots, 2m\}$ or we want to update an arbitrary integer $T_i$ for some $i \in \{0, \ldots, 2m\}$. Our goal is to implement the queries and updates in $\Oh(\log m)$ time. This indeed can be done by a simple binary tree.\footnote{The bounds are known to be tight up to a $\log \log m$ factor in the cell-probe model~\cite{patrascu2006logarithmic}. In particular, if the query (update) time is $\log^{\Oh(1)}m$, then the update (query) time is $\Omega(\log(m)/\log\log m)$.} Such a data structure implies that we can check the condition on line 2 in $\Oh(\log m)$ time. To bound the final running time, consider a particular position where $S^{i-1}$ and $S^{i-1}+x_i$ differ. This position can cause the condition $(S+x) \cap [a,b) = S \cap [a,b)$ to fail (and the algorithm to proceed to line 3) at most $\Oh(\log m)$ times. Each time we spend $\Oh(\log m)$ time to check the condition and the total number of positions where $S^{i-1}$ and $S^{i-1}+x_i$ differ is $2|C^i|$. In total, this implies that the function outputs $C^i$ in time $\Oh(\log m) \cdot \Oh(\log m) \cdot 2|C^i| = \Oh(|C^i|\log^2 m)$, assuming that $|C^i| > 0$. If $C^i = \emptyset$ the running time is $\Oh(\log m)$. Finally, we observe that we can perform the Bellman iteration $S^i = S^{i-1}\cup C^i$ in $\Oh(|C^i|\log m)$ time.

Combining the above, we arrive at our first result. 

\begin{theorem}\label{thm:result_1}
	Modular Subset Sum can be solved in $\Oh(m \log^2 m)$ time with high probability.
\end{theorem}

A sample implementation in Python is given in Appendix~\ref{app:implementation}. It uses a simple and efficient implementation of binary trees for maintaining prefix sums~\cite{fenwick1994new}.

\section{Algorithm II: Dynamic Strings}
\label{sec:algotwo}

This section is devoted to the second algorithm for Modular Subset Sum. In particular, we prove the following theorem. Recall that we can assume $n = \Oh(|X^\ast|)$ (after an $\Oh(n \log n)$-time preprocessing).

\begin{theorem}\label{thm:result_2}
Modular Subset Sum can be solved by a deterministic algorithm in time $\Oh(|X^\ast| \,\polylog\, |X^\ast|)$, where $X^*$ denotes the set of attainable subset sums of $X$ modulo $m$.
\end{theorem}

We first set up the necessary notation on strings.  A string $z$ of length $|z|$ is a sequence of letters from alphabet $\Sigma$ referred to as $z[0],\ldots,z[|z|-1]$. By $z[i..j]$ we denote the substring from letter $z[i]$ up to letter $z[j]$. We write $z[..j]$ as shorthand for $z[0..j]$ and similarly $z[i..]$ for $z[i..|z|-1]$.

\subsection{Data Structure for Dynamic Strings}

We start by reviewing a classic tool in string algorithms. 
This is a data structure for efficiently maintaining a family $\cal F$ of strings over alphabet $\Sigma$ under the following update operations. 
\begin{itemize}
\item \textsc{AddString}$(c)$: Given a letter $c \in \Sigma$, this operation adds the 1-letter string $c$ to $\cal F$.
\item \textsc{Concatenate}$(s,s')$: Given strings $s,s' \in \cal F$, concatenate them and add the resulting string to $\cal F$. The two strings $s,s'$ remain in $\cal F$.
\item \textsc{Split}$(s,i)$: Given a string $s \in \cal F$ and a number $i$, split $s$ into two strings $s[..i-1]$ and $s[i..]$ and add these strings to $\cal F$. The string $s$ remains in $\cal F$.
\item \textsc{Equal}$(s,s')$: Given strings $s,s' \in \cal F$, return true if $s = s'$.
\end{itemize}

Note that no string is ever removed from $\cal F$.
Mehlhorn et al.~\cite{MehlhornSU97} were the first to design a data structure supporting these operations in polylogarithmic time. Their time bounds have been further improved~\cite{AlstrupBR98,AlstrupBR00,GawrychowskiKKL18}, but since we will ignore logarithmic factors we shall not make use of those improvements.

\begin{theorem}[\cite{MehlhornSU97}] \label{thm:mehlhorn}
  There is a deterministic data structure for maintaining a family of strings under the operations \textsc{AddString}, \textsc{Concatenate}, \textsc{Split}, and \textsc{Equal} such that any sequence of $k$ operations resulting in total size $N = \sum_{s \in \cal F} |s|$ runs in time $\Oh(k \,\polylog (kN))$.
\end{theorem}

The data structure even works for very large alphabet $\Sigma$, as long as $\Sigma$ is ordered and we can compare any two letters in time $\Oh(1)$.

We observe that as an application of the above we obtain the following data structure.
\begin{lemma} \label{lem:ds}
  There is a deterministic data structure that maintains a length-$m$ string $z$ over alphabet $\{0,1\}$, initialized as $z = 0^m$, under the following operations, where any sequence of $k$ operations runs in time $\Oh(k \,\polylog(km))$:
  \begin{itemize}
    \item \textsc{Add}$(i)$: Given $0 \le i < m$, set $z[i] := 1$,
    \item \textsc{LCP}$(i,j)$: Return the length of the longest common prefix of $z[i..]$ and $z[j..]$. %
  \end{itemize}
\end{lemma}
\begin{proof}
  For the initialization of $z = 0^m$, we first run \textsc{AddString}$(0)$ and then, using $\Oh(\log m)$ concatenations, we generate the strings $0^{2^i}$ and we combine them according to the binary representation of $m$ to obtain the string $0^m$.
  
  For \textsc{Add}$(i)$, we split $z$ at $i$ and at $i+1$ to obtain the strings $z[..i-1]$ and $z[i+1..]$. We then run \textsc{AddString}$(1)$, and finally we concatenate twice to obtain the resulting string $z' = \textsc{Concatenate}(\textsc{Concatenate}(z[..i-1], 1), z[i+1..])$.
  
  For a longest common prefix query \textsc{LCP}$(i,j)$, we first split $z$ at $i$ and at $j$ to obtain the strings $y_1 := z[i..]$ and $y_2 := z[j..]$. Then we perform a binary search for the largest $\ell$ such that $y_1[..\ell] = y_2[..\ell]$. Each step of the binary search uses two splits, to construct the strings $y_1[..\ell]$ and $y_2[..\ell]$, and one equality test.
  
  Hence, we can simulate $k$ operations among \textsc{Add} and \textsc{LCP} using $\Oh(k \log m)$ operations among \textsc{Equal}, \textsc{AddString}, \textsc{Concatenate}, and \textsc{Split}. The total string length of the constructed family $\cal F$ is $N = \Oh(m k \log m)$, and thus the total time is $\Oh(k \log m \,\polylog (kN)) = \Oh(k \,\polylog(km))$.
\end{proof}

\subsection{From Dynamic Strings to Modular Subset Sum}

Recall that given $X = \{x_1,\ldots,x_n\} \subseteq \mathbb{Z}_m$ our aim is to compute the set $X^\ast \subseteq \mathbb{Z}_m$ consisting of all subset sums of~$X$. 
As in our first algorithm for Modular Subset Sum, we follow Bellman's approach by initializing $S^0 := \{0\}$ and iteratively computing $S^i := (S^{i-1} + x_i) \cup S^{i-1}$ for $i=1,\ldots,n$. As we have seen in the first algorithm, it suffices to compute the symmetric difference $E^i := (S^{i-1} + x_i) \triangle S^{i-1}$ in time $\Oh((|E^i|+1) \,\polylog\, m)$; then over all iterations we compute $S^n = X^\ast$ in time $\Oh(|X^\ast| \,\polylog\, m)$.

It remains to show how to compute the symmetric difference $E^i$ in each iteration.
To this end, let $z$ be the indicator vector of $S^i$ copied twice, that is, $z$ is a string of length $2m$ over alphabet $\{0,1\}$ where $z[j]$ indicates whether $j \bmod m$ is in $S^i$, for any $0 \le j < 2m$. 
We maintain the data structure from Lemma~\ref{lem:ds} for the string $z$. Since this data structure initializes $z$ as $0^{2m}$, we call \textsc{Add}$(0)$ and \textsc{Add}$(m)$ to initialize $z$ correctly according to $S^0 = \{0\}$. 

At the beginning of the $i$-th iteration, note that $z[m-x_i..2m-x_i]$ is the indicator vector of $S^{i-1} + x_i$.
The query \textsc{LCP}$(0,m-x_i)$ yields a number $d'$ such that $d := d'+1$ is minimal with $z[d] \ne z[m-x_i+d]$. In other words, $d$ is the smallest element of the symmetric difference $E^i$ of $S^{i-1}$ and $S^{i-1} + x_i$ (unless $d \ge m$, in which case we have that $E^i$ is the empty set). We find the next element of $E^i$ by calling \textsc{LCP}$(d+1,m-x_i+d+1)$. Repeating this argument, we compute the set $E^i$ using $\Oh(|E^i|+1)$ \textsc{LCP} operations. This finishes the description of how to compute the symmetric difference $E^i$.
We maintain the string $z$ for the next iteration by setting $z[d] = z[d+m] = 1$ for each $d \in E^i$ with $d \not\in S^{i-1}$. This uses $\Oh(|E^i|)$ \textsc{Add} operations.

In total, we run $\Oh(|X^*|+n) = \Oh(|X^*|)$ \textsc{LCP} and \textsc{Add} operations on a string of length $2m$. This takes total time $\Oh(|X^*| \,\polylog\, (|X^*| \,m)) = \Oh(|X^*| \,\polylog\, m)$ according to Lemma~\ref{lem:ds}. In particular, this running time is bounded by $\tOh(m)$. We further improve this running time in Section~\ref{sec:improverunningtime} below. For pseudocode see Algorithm~\ref{alg:main}.

\begin{algorithm}[ht]
\caption{Algorithm for Modular Subset Sum using dynamic strings.}
\label{alg:main}
\begin{algorithmic}[1]
\Function{ModularSubsetSumViaDynamicStrings}{$X,m$}
\State $S := \{0\}$
\State Initialize $z = 0^{2m}$ (as in Lemma~\ref{lem:ds})
\State $z.\textsc{Add}(0)$
\State $z.\textsc{Add}(m)$
\For {$i=1,\ldots,n$}
  \State $E^i := \emptyset$
  \State $d := 1 + z.\textsc{LCP}(0, m-x_i)$
  \While {$d < m$}
    \State $E^i := E^i \cup \{d\}$
    \State $d := d + 1 + z.\textsc{LCP}(d+1, m-x_i+d+1)$
  \EndWhile
  \For {\textbf{each} $d \in E^i$} 
    \If {$d \not\in S$} 
      \State $S := S \cup \{d\}$
      \State $z.\textsc{Add}(d)$
      \State $z.\textsc{Add}(d+m)$
    \EndIf
  \EndFor
\EndFor
\State \Return $S$
\EndFunction
\end{algorithmic}
\end{algorithm}

\subsection{Solution Reconstruction}

In order to reconstruct a subset $Y \subseteq X$ summing to a given target $t$, we augment the above algorithm as follows. We store the set $S^i$ in a balanced binary search tree~${\cal T}^i$. For each number $d \in S^i \setminus S^{i-1}$, in the node corresponding to $d$ in ${\cal T}^i$ we store a pointer to the node corresponding to $d-x_i$. At the end of the algorithm ${\cal T}^n$ stores $S^n = X^*$, the set of all subset sums of $X$. Note that computing ${\cal T}^n$ augmented by these pointers takes total time $\Oh(|X^*| \log |X^*|)$ and thus does not increase the asymptotic running time of the algorithm.

With this bookkeeping, given any target integer $t \in \mathbb{Z}_m$ we first search for $t$ in ${\cal T}^n$ to check whether $t \in X^*$. If $t \in X^*$, then starting from the node corresponding to $t$ in ${\cal T}^n$, we follow the stored pointers to reconstruct a subset $Y \subseteq X$ summing to $t$ modulo $m$. The total running time of this solution reconstruction is $\Oh(|Y| + \log |X^*|)$. 

Clearly, we have $|Y| \le n$. This is essentially the only control we have over the size $|Y|$, in particular we do not guarantee~$Y$ to be a smallest subset summing to $t$.

\subsection{Improving the Running Time}
\label{sec:improverunningtime}

We now improve the running time from $\Oh(|X^*| \,\polylog\, m)$ to $\tOh(|X^*|) = \Oh(|X^*| \,\polylog\, |X^*|)$, finishing the proof of Theorem~\ref{thm:result_2}. Observe that all steps of the algorithm (including the solution reconstruction) run in time $\tOh(|X^*|)$, except for Lemma~\ref{lem:ds}. Hence, it suffices to replace this lemma by the following improved variant, which makes use of run-length encoding. 

\begin{lemma} \label{lem:dsimproved}
  There is a deterministic data structure that can initialize $z = 0^m$ and perform $k$ \textsc{Add} and \textsc{LCP} operations in total time $\Oh(k \,\polylog\, k)$.
\end{lemma}

\begin{proof}
  Recall that we assume that arithmetic operations on $\Oh(\log m)$-bit numbers can be performed in time $\Oh(1)$. In particular, the string length $m$ can be processed in time $\Oh(1)$.

  Denote by $S \subseteq \mathbb{Z}_m$ the set of which $z$ is the indicator vector. That is, initially we have $S = \emptyset$ and on operation \textsc{Add}$(i)$ we update $S := S \cup \{i\}$. 
  We store $S$ in a balanced binary search tree~$\cal T$. 
  We also augment $\cal T$ to store at each node the size of its subtree. This allows us to perform the following queries in time $\Oh(\log |S|)$: 
  \begin{itemize}
  \item \emph{Rank:} Given a number $v$, determine the number of keys stored in $\cal T$ that are smaller than $v$, 
  \item \emph{Select:} Given a number $v$, determine the $v$-th number stored in $\cal T$ (in sorted order).
  \end{itemize}
  Note that $\cal T$ can be updated in time $\Oh(\log |S|)$ per operation. The total time for maintaining $\cal T$ during $k$ \textsc{Add} and \textsc{LCP} operations is $\Oh(k \log k)$, since $|S| \le k$.
  
  We compress the string $z$ by replacing each run of 0's by one symbol. Specifically, let $\Sigma := \{1\} \cup \{(0,L) \mid 0 \le L \le m\}$. Note that symbols in $\Sigma$ can be read and compared in time $\Oh(1)$. We convert string $z \in \{0,1\}^m$ to a string $C(z) \in \Sigma^*$ by replacing each maximal substring $0^L$ of $z$ by the symbol $(0,L)$. For simplicity, we also add the symbol $(0,0)$ between any two consecutive 1's in~$z$. For example, the string $z = 10001100$ is converted to $C(z) = 1 (0,3) 1 (0,0) 1 (0,2)$.
  We use the data structure of Theorem~\ref{thm:mehlhorn} to store $C(z)$. We maintain $C(z)$ using the binary search tree~$\cal T$, by implementing initialization, \textsc{Add}, and \textsc{LCP} as follows.
  
  \smallskip
  \emph{Initialization.} Given $m$, we initialize $z = 0^m$ and thus $C(z) = (0,m)$. This string is generated by calling \textsc{AddString}$(c)$ for $c = (0,m) \in \Sigma$, which takes time $\Oh(1)$.
  
  \smallskip
  \emph{Add.} Given $i$, we want to set $z[i] := 1$. 
  We denote by $a < i < b$ the predecessor and successor of $i$ in $S$, so that $z[a..b] = 1 0^{b-a-1} 1$. Note that $a$ and $b$ can be computed from $\mathcal{T}$. Using a rank query on $i$, we can infer the corresponding position $h$ with $C(z)[h-1..h+1] = 1 (0,b-a-1) 1$. We split $C(z)$ at $h$ and at $h+1$ to obtain the strings $C(z)[..h-1]$ and $C(z)[h+1..]$. We then construct the string $(0,i-a-1) 1 (0,b-i-1)$ using \textsc{AddString} thrice and \textsc{Concatenate} twice. Finally, we concatenate $C(z)[..h-1]$ and $(0,i-a-1) 1 (0,b-i-1)$ and $C(z)[h+1..]$ to form the new string $C(z)$ after setting $z[i] := 1$. 
  
  \smallskip
  \emph{LCP.} Given $i,j$, let $y_1 := z[i..]$ and $y_2 := z[j..]$. We first construct the strings $C(y_1)$ and $C(y_2)$. This is similar to the last paragraph: 
  Denote the predecessor and successor of $i$ by $a < i \le b$, so that $z[a..b] = 1 0^{b-a-1} 1$. Using a rank query, we find the corresponding position $h$ with $C(z)[h-1..h+1] = 1 (0,b-a-1) 1$. Splitting $C(z)$ at $h+1$ and concatenating it after $(0,b-i-1)$ yields $C(y_1)$. (If $b-i-1 = 0$ then we remove the initial $(0,0) = (0,b-i-1)$.) We similarly generate $C(y_2)$. 
  We now perform a binary search for the largest $\ell$ such that $C(y_1)[..\ell] = C(y_2)[..\ell]$, using two \textsc{Split} and one \textsc{Equal} operation per binary search step. We use a rank and a select query to determine the length $\Delta$ of the string corresponding to $C(y_1)[..\ell]$, that is, $C(z[i..i+\Delta-1]) = C(y_1)[..\ell]$. If $C(y_1)[\ell+1] = 1$ or $C(y_2)[\ell+1] = 1$ or one of these symbols is undefined (i.e., out of bounds) then \textsc{LCP}$(i,j) = \Delta+1$. Otherwise, we have $C(y_1)[\ell+1] = (0,L_1)$ and $C(y_2)[\ell+1] = (0,L_2)$, and then \textsc{LCP}$(i,j) = \Delta + \min\{L_1,L_2\} + 1$.
  
  \smallskip
  Hence, we can simulate $k$ operations among \textsc{Add} and \textsc{LCP} using $\Oh(k \log k)$ operations among \textsc{Equal}, \textsc{AddString}, \textsc{Concatenate}, and \textsc{Split}. The total string length of the constructed family~$\cal F$ is $N = \Oh(k^2 \log k)$, since after $k$ operations each constructed string has at most $k$ 1's and thus has length $\Oh(k)$. 
  By Theorem~\ref{thm:mehlhorn}, the total time is $\Oh(k \,\polylog (kN)) = \Oh(k \,\polylog\, k)$.
\end{proof}

\bibliography{main}

\appendix

\section{Simple and Fast Algorithm for All-Pairs Non-Decreasing Paths}
\label{app:apnp}

In the APNP problem, given an edge-weighted graph, the goal is to compute for any pair of nodes $a$ and $b$ the minimum cost of a path from $a$ to $b$ that uses non-decreasing edge-weights. The cost of such a path is defined to be the largest edge-weight encountered on the path.

There has been a number of works that sequentially improved the running time for the directed and undirected case of APNP~\cite{vassilevska2008nondecreasing,duan2018improved,DuanJW19}. The directed case is a generalization of the max-min matrix product~\cite{vassilevska2007all,duan2009fast} and the best known algorithm for both problems runs in time $\tOh(n^{(\omega+3)/2})$, where $\omega$ is the exponent of fast matrix multiplication~\cite{coppersmith1990matrix,williams2012multiplying,le2014powers}. In contrast, the undirected case is known to be solvable in $\tOh(n^2)$ time~\cite{DuanJW19}. %

We show how to solve the undirected APNP problem by a simple algorithm in time $\Oh(n^2 \log n)$. This improves the previously best result in terms of log-factors, and it is optimal up to a single log-factor. For simplicity, in the following we call the undirected case of APNP simply APNP.
In this section we prove the following theorem.

\begin{theorem}
	All-Pairs Non-Decreasing Paths can be solved in $\Oh(n^2 \log n)$ time w.h.p.
\end{theorem}

Let $G=(V,E)$ be an undirected graph with $n = |V|$ nodes and $m = \Oh(n^2)$ edges having edge weights $w(e)$ for $e \in E$. A path is a sequence of edges $e_1, e_2, \ldots, e_\ell$, such that $e_i,e_{i+1}$ share an endpoint for all $1 \le i\le \ell - 1$.
 A non-decreasing path is a path satisfying $w(e_i)\leq w(e_{i+1})$ for all $1\leq i \leq \ell-1$. The weight of this non-decreasing path is defined to be $w(e_\ell)$, the weight of the last edge. The All Pairs Non-Decreasing Paths Problem (APNP) asks to determine the minimum weight non-decreasing path between every pair of vertices.

For simplicity, we focus on the strictly increasing version of the problem where there are no edges of equal weight. The general case can be converted to the distinct weights case (see Lemma 23 in Duan et al.~\cite{duan2019faster}), through a simple reduction. The reduction looks for connected components formed by edges of the same weight and replaces these edges with new ones with distinct weights. This preprocessing step runs in $\Oh(n^2)$ time and the number of edges in the new graph at most doubles. It thus suffices to focus on the distinct weights case.

The algorithm starts by ordering all edges of the graph from the smallest weight to the largest and inspecting the edges in this order. For every vertex $u$ of the graph we maintain a set of vertices $v$ that can be reached from $u$ by a non-decreasing path using only the edges that have been inspected so far. Initially the sets for all vertices are empty. The first time a vertex $v$ is added to a list corresponding to a vertex $u$ determines the cost of minimum non-decreasing path from $u$ to $v$. In particular, if the vertex $v$ is added to the list corresponding to the vertex $u$ in the phase when we are inspecting edge $e$, the weight of the minimum non-decreasing path from $u$ to $v$ is equal to the weight $w(e)$ of the edge $e$. Let $C_e$ be the set of newly discovered reachability pairs added in the phase when inspecting the edge $e$. We will shortly describe how we can compute $C_e$ in an output-sensitive time $\Oh(|C_e|\log n + 1)$. This implies that the total running time is upper bounded by $\Oh(n^2 \log n)$ since the total number of node pairs is upper bounded by $n^2$ and each pair is discovered at most once.

Now we describe how to compute $C_e$ in $\Oh(|C_e|\log n + 1)$ time. Let $e = (a,b)$. Let $u$ be a vertex that can reach $a$ but cannot reach $b$ only using the edges inspected so far (not including $e$) via a non-decreasing path. We observe that, by adding the edge $e=(a,b)$, the vertex $u$ can now reach vertex $b$ (by first going to $a$ and then traversing the edge $e$). Similarly, if $u$ can reach $b$ but cannot reach $a$, after adding $e$, $u$ can reach $a$. On the other hand, if $u$ can reach both $a$ and $b$ (or cannot reach both), no new edges will be added from $u$ after inspecting edge $e$. Let $R^a$ be the set of vertices $u$ that can reach $a$ but cannot reach $b$ (right before inspecting $e$), and $R^b$ be the set of vertices that can reach $b$ but not $a$. We conclude that $C_e = ((R^a \setminus R^b) \times \{b\}) \cup ((R^b \setminus R^a) \times \{a\})$. Therefore it is sufficient to be able to compute $R^a \setminus R^b$ and $R^b \setminus R^a$ in $\Oh(|C_e|\log n + 1)$ time. If we spend $\Oh(\log n)$ time per single vertex from one of these two sets, we obtain the required running time. We use a similar idea as we used for Modular Subset Sum. Let $R^a_i=1$ if the $i$-th vertex of the graph belongs to $R^a$ and $R^a_i=0$ otherwise. For a random integer $r \in \{0, \ldots, p-1\}$ (for a large enough prime $p$) we build a tree data structure that stores partial sums of the sequence $R^a_0 \cdot r^0, R^a_1 \cdot r^1, R^a_2 \cdot r^2, \ldots$ in its internal leaves. In particular, we associate the leaves of a complete binary tree with the elements of the sequence and each node recursively stores the sum of values of its children. We can update an element of the sequence by spending $\Oh(\log n)$ time on the data structure. Furthermore, if we have data structures for $R^a$ and $R^b$, we can recursively inspect subtrees (whose hash values disagree) of the two data structures to find all elements from $R^a \setminus R^b$ and $R^b \setminus R^a$. The time spent to find one element is $\Oh(\log n)$. Thus, if we store such a data structure for each vertex of the graph, we can update them efficiently and compute $C_e$ in time $\Oh(|C_e|\log n + 1)$ for any edge $e$.

A sample implementation in Python is given in Appendix~\ref{app:apnpimplementation}.

\section{Python Implementation of Modular Subset Sum} \label{app:implementation}

Below we present a simple implementation of our first algorithm for Modular Subset Sum (Theorem~\ref{thm:result_1}) in Python.\footnote{The code can be also found at \url{https://ideone.com/YlLwMQ}.}
It maintains a binary indexed tree that keeps track of the prefix sums of polynomial hashes of the characteristic vector of the attainable subsets.
To easily deal with rollover due to the cyclicity of the mod operation, a separate copy of the characteristic vector is kept translated by $m$.

It takes as an input a list of numbers $W$ and the modulus $m$, and returns a list of length $m$, where the entry at position $s$ is None if $s$ is not a 
possible subset sum of $W$ modulo $m$, or contains the last number from $W$ that was added to create the subset sum $s$.

\begin{python}[mathescape]
import random

def ModularSubsetSum(W, m):
    p = 1234567891                            #large prime ${\color{commentcolour} p > m^2 \log m}$
    r = random.randint(0,p)                   #random number ${\color{commentcolour} r}$ in [0,p)
    powr = [1]                                #Precompute powers of ${\color{commentcolour} r }$
    for i in range(2*m):                      #powr[i] ${\color{commentcolour} \triangleq }$ ${\color{commentcolour} r^i }$ (mod p)
        powr.append((powr[-1] * r) % p)
    
    
    #Binary Indexed Tree for prefix sums
    tree = [0] * (2*m)
    def read(i):                              #Prefix sum of [0,i)
        if i<=0: return 0
        return tree[i-1] + read(i-(i&-i)) 
    def update(i, v):                         #add v to position i
        while i < len(tree):
            tree[i] += v
            i += (i+1)&-(i+1)
    
    
    #Functions for finding new subset sums and adding them
    def FindNewSums(a,b,w):
        h1 = (read(b)-read(a))*powr[m-w] % p  #hash of ${\color{commentcolour} S \cap [a,b)}$
        h2 = (read(b+m-w)-read(a+m-w)) % p    #hash of ${\color{commentcolour} (S+w) \cap [a,b)}$
        if h1 == h2: return []
        if b == a+1:                        
            if sums[a] is None: return [a]    #a is a new sum
            return []                         #a is a ghost sum
        return FindNewSums(a,(a+b)//2,w) + FindNewSums((a+b)//2,b,w)
    def AddNewSum(s, w):
        sums[s] = w
        update(s,powr[s]), update(s+m,powr[s+m])
    
    
    #Main routine for computing subset sums
    sums = [None] * m
    AddNewSum(0,0)
    for w in W:
        for s in FindNewSums(0,m,w):
            AddNewSum(s,w)
            
    return sums

\end{python}

\subsection*{Example} 
Find all modular subset sums mod 8 with numbers 1, 3 and 6:
\begin{python}
ModularSubsetSum([1,3,6], 8)  #Returns [0, 1, 6, 3, 3, None, 6, 6]
\end{python}

\subsection*{Recovering the subset} 
To recover the subset making a particular subset sum, we repeatedly subtract the last number added in the subset sum $s$ until we get down to 0.
\begin{python}
def RecoverSubset(sums, s):
    if sums[s] is None: return None
    if s <= 0: return []
    return RecoverSubset(sums, (s-sums[s]) % len(sums)) + [ sums[s] ]
\end{python}

\begin{python}
sums = ModularSubsetSum([1,3,6], 8)
RecoverSubset(sums, 7)  #Returns [1, 6]
RecoverSubset(sums, 2)  #Returns [1, 3, 6]
\end{python}

\section{Python Implementation of All-Pairs Non-Decreasing Paths}\label{app:apnpimplementation}

Below we present a simple implementation of our algorithm in Python for computing minimum weight non-decreasing path between all pairs of $n$ vertices.\footnote{The code can be also found at \url{https://ideone.com/S9RAhX}.} 
It takes as an input a list $E$ of edges of the graph in increasing order of their weights and the number $n$ of vertices. Note that the actual weights of the edges do not matter besides their relative order. The algorithm returns an $n \times n$ matrix \textsf{path}. \textsf{path}$[u,v]=$ None if there is no way to reach $v$ from $u$ by traversing edges with increasing weights. Otherwise \textsf{path}$[u,v]=$ \textsf{par}, where \textsf{par} is the previous vertex on the minimum weight non-decreasing path from $u$ to $v$. For every vertex of the graph the algorithm keeps track of partial hashes of vertices that can reach this vertex in a tree data structure.

\iffalse
----------

It maintains a binary indexed tree that keeps track of the prefix sums of polynomial hashes of the characteristic vector of the attainable subsets.
To easily deal with rollover due to the cyclicity of the mod operation, a separate copy of the characteristic vector is kept translated by $m$.

It takes as input a list of numbers $W$ and the modulus $m$, and returns a list of length $m$, where the entry at position $s$ is None if $s$ is not a 
possible subset sum modulo $m$, or contains the last number from $W$ that was added to create the subset sum $s$.

\fi

\begin{python}[mathescape]
import random

def AllPairsNonDecreasingPaths(E, n):
    p = 1234567891                            #large prime ${\color{commentcolour} p > n^3 \log n}$
    r = random.randint(0,p)                   #random number ${\color{commentcolour} r}$ in [0,p)
    powr = [1]                                #Precompute powers of ${\color{commentcolour} r }$
    for i in range(n):                        #powr[i] ${\color{commentcolour} \triangleq }$ ${\color{commentcolour} r^i }$ (mod p)
        powr.append((powr[-1] * r) % p)
    
    N = 1<<(n-1).bit_length()                 #round n to next power of 2
    tree = [ [0]*(2*N) for _ in range(n) ]
    def update(v,node,val):
        while node > 0:
            tree[v][node] += val
            node >>= 1
      
    #Functions for finding new paths and adding them
    def FindNewPaths(a,b,node):
        if tree[a][node] == tree[b][node]: return []
        if node >= N:                         
            u = node - N                      #leaf node
            if path[u][a] is None: 
                return [(u,a,b)]
            return [(u,b,a)]
        return FindNewPaths(a,b,2*node) + FindNewPaths(a,b,2*node+1)
    def AddNewPath(u, v, par):
        path[u][v] = par
        update(v,u+N,powr[u])
    
    #Main routine for finding all pairs non-decreasing paths
    path = [ [None] * n for _ in range(n) ]
    for i in range(n):
        AddNewPath(i,i,i)
    for (a,b) in E:
        for (u,v,par) in FindNewPaths(a,b,1):
            AddNewPath(u,v,par)
            
    return path
\end{python}

\subsection*{Example} 
Find all pairs non-decreasing paths in a graph with $n$ nodes and edges $(1,2)$ and $(0,1)$:
\begin{python}
AllPairsNonDecreasingPaths([(1,2),(0,1)], 3)  
                #Returns [[0, 0, None], [1, 1, 1], [1, 2, 2]]
\end{python}

\subsection*{Recovering the path between two vertices} 
To recover a specific path between two vertices, we repeatedly move to the last node visited before the destination until we reach the source.
\begin{python}
def RecoverPath(path,u,v):
    if path[u][v] is None: return None
    if u == v: return [u]
    return RecoverPath(path, u, path[u][v]) + [ v ]
\end{python}
\begin{python}
path = AllPairsNonDecreasingPaths([(1,2),(0,1)], 3)
RecoverPath(path, 2, 0)  #Returns [2,1,0]
RecoverPath(path, 0, 2)  #Returns None
\end{python}

\end{document}